\documentclass[a4paper,11pt]{scrartcl}

\usepackage{fullpage}

\usepackage{amsmath,amsthm,amssymb}
\usepackage{xspace}

\theoremstyle{plain}
\newtheorem{proposition}{Proposition}
\newtheorem{theorem}{Theorem}

\theoremstyle{definition}
\newtheorem{definition}{Definition}

\newcommand{\yes}{\textbf{yes}\xspace}
\newcommand{\no}{\textbf{no}\xspace}

\newcommand{\OR}{\textsc{or}\xspace}

\newcommand{\vertexcover}{\textsc{Vertex Cover}\xspace}
\newcommand{\independentset}{\textsc{Independent Set}\xspace}
\newcommand{\clique}{\textsc{Clique}\xspace}
\newcommand{\hittingset}{\textsc{Hitting Set}\xspace}
\newcommand{\dhittingset}{$d$-\textsc{Hitting Set}\xspace}
\newcommand{\ILPFn}{\textsc{Integer Linear Program Feasibility($n$)}\xspace}
\newcommand{\ILPF}{\textsc{Integer Linear Program Feasibility}\xspace}
\newcommand{\ilpfn}{\textsc{ILPF($n$)}\xspace}
\newcommand{\RSILPF}{$r$-\textsc{Sparse Integer Linear Program Feasibility}\xspace}
\newcommand{\RSILPFn}{$r$-\textsc{Sparse Integer Linear Program Feasibility($n$)}\xspace}
\newcommand{\rsilpf}{$r$-\textup{SILPF}\xspace}
\newcommand{\rsilpfn}{$r$-\textup{SILPF($n$)}\xspace}
\newcommand{\RSBILPFnd}{$r$-\textsc{Sparse Bounded Integer Linear Program Feasibility($n$,$d$)}\xspace}

\newcommand{\sat}{\textup{SAT}\xspace}
\newcommand{\dsat}{$d$-\textup{SAT}\xspace}

\newcommand{\containment}{\textup{NP}~$\subseteq$~\textup{coNP/poly}\xspace}
\newcommand{\noncontainment}{\textup{NP}~$\nsubseteq$~\textup{coNP/poly}\xspace}

\renewcommand{\P}{\mathcal{P}}
\newcommand{\R}{\mathcal{R}}
\newcommand{\Oh}{\mathcal{O}}

\renewcommand{\i}{{i^*}}

\newcommand{\N}{\ensuremath{\mathbb{N}}\xspace}
\newcommand{\Q}{\ensuremath{\mathbb{Q}}\xspace}
\newcommand{\Z}{\ensuremath{\mathbb{Z}}\xspace}

\newcommand{\cref}[1]{(\ref{#1})\xspace}

\newcommand{\parameterizedproblem}[4]{%
\medskip
\noindent\fbox{\begin{minipage}{0.985\linewidth}
  #1 \\
  \textbf{Input:} #2 \\
  \textbf{Parameter:} #3\\ 
  \textbf{Output:} #4
\end{minipage}}
\medskip
}

\renewcommand{\subparagraph}[1]{\noindent\textbf{#1}}

\title{On Polynomial Kernels for Sparse Integer Linear Programs}

\author{Stefan Kratsch\thanks{Technical University Berlin, Germany, \texttt{stefan.kratsch@tu-berlin.de}}~\thanks{Work done in part at Utrecht University, the Netherlands, supported by the Netherlands Organization for Scientific Research (NWO), project ``KERNELS'', and in part at Max-Planck-Institute for Informatics, Saarbr\"ucken, Germany.}}
\date{}

\begin{document}

\maketitle

\begin{abstract}
Integer linear programs (ILPs) are a widely applied framework for dealing with combinatorial problems that arise in practice. It is known, e.g., by the success of CPLEX, that preprocessing and simplification can greatly speed up the process of optimizing an ILP. The present work seeks to further the theoretical understanding of preprocessing for ILPs by initiating a rigorous study within the framework of parameterized complexity and kernelization.

A famous result of Lenstra (Mathematics of Operations Research, 1983) shows that feasibility of any ILP with~$n$ variables and~$m$ constraints can be decided in time~$\Oh(c^{n^3}\cdot m^{c'})$. Thus, by a folklore argument, any such ILP admits a kernelization to an equivalent instance of size~$\Oh(c^{n^3})$. It is known, that unless \containment and the polynomial hierarchy collapses, no kernelization with size bound polynomial in~$n$ is possible. However, this lower bound only applies for the case when constraints may include an arbitrary number of variables since it follows from lower bounds for \sat and \hittingset, whose bounded arity variants admit polynomial kernelizations.

We consider the feasibility problem for ILPs~$Ax\leq b$ where~$A$ is an~$r$-row-sparse matrix parameterized by the number of variables. We show that the kernelizability of this problem depends strongly on the range of the variables. If the range is unbounded then this problem does not admit a polynomial kernelization unless \containment. If, on the other hand, the range of each variable is polynomially bounded in~$n$ then we do get a polynomial kernelization. Additionally, this holds also for the more general case when the maximum range~$d$ is an additional parameter, i.e., the size obtained is polynomial in~$n+d$.
\end{abstract}

\section{Introduction}

The present work seeks to initiate a study of the preprocessing properties of integer linear programs (ILPs) within the framework of parameterized complexity. Generally, preprocessing (or data reduction) is a universal strategy for coping with combinatorially hard problems and can be combined with other strategies like approximation, brute-force, exact exponential-time algorithms, local search, or heuristics. Unlike those other approaches, preprocessing itself incurs only a polynomial-time cost and is error free (or, in rare cases, with negligible error); recall that under standard assumptions we do not expect to exactly solve any NP-hard problem in polynomial time. Thus, preprocessing before applying other paradigms is essentially free and saves solution quality and/or runtime on parts of the input that are sufficiently easy to handle in polynomial time (see e.g.~\cite{Weihe1998_trains}). For a long time, preprocessing has been neglected in theoretical research for lack of appropriate tools\footnote{In fact, it has been observed that no polynomial-time algorithm can shrink all instances of some NP-hard problem unless P~$=$~NP~\cite{HarnikN10}; this issue can be avoided in parameterized complexity.} and research was limited to experimental evaluation of preprocessing strategies. The introduction of parameterized complexity and its notion of kernelization has sparked a strong interest in theoretically studying preprocessing with proven upper and lower bounds on its performance.

Integer linear programs are widely applied in theory and practice. There is a huge body of scientific literature on ILPs both as a topic of research itself and as a tool for solving other problems. From a theoretical perspective, many fundamental problems that revolve around ILPs are hard, e.g., checking feasibility of a $0/1$-ILP is NP-hard by an easy reduction from the classic \textsc{Satisfiability} problem~\cite{GareyJ79}. Similarly, it is easy to express \vertexcover or \independentset, thus showing that simple covering and packing ILPs are NP-hard to optimize. Thus, for worst-case complexity considerations, the high expressive power of ILPs comes at the price of encompassing plenty of hard problems and, effectively, inheriting all their lower bounds (e.g., approximability).

In practice, the expressive power of ILPs makes them a versatile framework for encoding and solving many combinatorially hard problems. Coupled with powerful software packages for optimizing ILPs this has created a viable way for solving many practical problems on real-world instances. We refer to a survey of Atamt\"urk and Savelsbergh~\cite{AtamturkS05} for an explanation of the capabilities of modern ILP solvers; this includes techniques such as probing and coefficient reduction. One of the most well-known solvers is the CPLEX package, which is, in particular, known for its extensive preprocessing options and parameters.\footnote{The interested reader is referred to the online documentation and manual of ILOG CPLEX 12.4 at \texttt{http://pic.dhe.ibm.com/infocenter/cosinfoc/v12r4/index.jsp} (see ``presolve'', ``preprocessing'').} It is known that appropriate preprocessing and simplification of ILPs can lead to strong improvements in running time, e.g., reducing the range of variables or eliminating them altogether, or reducing the number of constraints. 
Given the large number of options that a user has for controlling the preprocessing in CPLEX, e.g., the number of substitution rounds to reduce rows and columns, this involves some amount of engineering and has a more heuristic flavor. In particular, there are no performance guarantees for the effect of the preprocessing.

Naturally, this leads to the question of whether there are theoretical performance guarantees for the viability of preprocessing for ILPs. To pursue this question in a rigorous and formal way, we take the perspective of parameterized complexity and its notion of (polynomial) kernelization. Parameterized complexity studies classical problems in a more fine-grained way by introducing one or more additional parameters and analyzing time- and space-usage as functions of input size and parameter. In particular, by formalizing a notion of \emph{fixed-parameter tractability}, which requires efficient algorithms when the parameter is small, this makes the parameter a quantitative indicator of the hardness of a given instance (see Section~\ref{section:preliminaries} for formal definitions). This in turn permits us to formalize preprocessing as a reduction to an equivalent instance of size bounded in the parameter, a so-called \emph{kernelization}. The intuition is that relatively easy instances should be reducible to a computationally hard, but small core, and we do not expect to reduce instances that are already fairly hard compared to their size (e.g., instances that are already reduced). While classically, no efficient algorithm can shrink \emph{each} instance of an NP-hard problem~\cite{HarnikN10}, the notion of kernelization has been successfully applied to a multitude of problems (see recent surveys by Guo and Niedermeier~\cite{GuoN07} and Bodlaender~\cite{Bodlaender09}). Due to many interesting upper bound results (e.g.,~\cite{BodlaenderFLPST09_meta,FominLST10,KratschW12_focs}) but also the fairly recent development of a lower bound framework for polynomial kernels \cite{HarnikN10,FortnowS11,BodlaenderDFH09,DellM10}, the existence or non-existence of polynomial kernels (which reduce to size polynomial in the parameter) is receiving high interest.

In this work, we focus on the effect that the dimension, i.e., the number of variables, has on the preprocessing properties of ILPs. Feasibility and optimization of ILPs with low dimension has been studied extensively already, see e.g.~\cite{Kannan80,Kannan1983,Lenstra1983,Megiddo1984,Kannan87,Clarkson1995,Eisenbrand03,EisenbrandS08}. The most important result for our purpose is a well-known work of Lenstra~\cite{Lenstra1983}, who showed that feasibility of an ILP with~$n$ variables and~$m$ constraints can be decided in time~$\Oh(c^{n^3}\cdot m^{\Oh(1)})$; this also means that the problem is fixed-parameter tractable with respect to~$n$. 
This has been improved further, amongst others by Kannan~\cite{Kannan87} to~$\Oh(n^{\Oh(n)})$ dependence on the dimension and by Clarkson~\cite{Clarkson1995} to (expected)~$\Oh((cn)^{n/2+\Oh(1)})$ dependence.
We take these results as our starting point and consider the problem of determining feasibility of a given ILP parameterized by the number of variables, formally defined as follows.

\parameterizedproblem{\ILPFn -- \ilpfn}{A matrix~$A\in \Q^{m\times n}$ and a vector~$b\in\Q^m$.}{$n.$}{Is there a vector~$x\in\Z^n$ such that~$Ax\leq b$?}

It is known by a simple folklore argument that any parameterized problem is fixed-parameter tractable if and only if it admits a kernelization; unfortunately the implied size guarantee is usually impractical as it is exponential in the parameter. As an example, using the runtime given by Kannan~\cite{Kannan87} we only get a kernel size of~$\Oh(n^{cn})$.\footnote{If the instance is larger than~$\Oh(n^{cn})$, then Kannan's algorithm runs in polynomial time and we may simply return the answer or a trivial \yes- or \no-instance. Otherwise, the claimed bound trivially holds.} Unsurprisingly, we are more interested in what kernel sizes can be achieved by nontrivial preprocessing rules. In particular, we are interested in the conditions under which an ILP with~$n$ variables can be reduced to size polynomial in~$n$, i.e., in the existence of polynomial kernels for \ILPFn.

\subparagraph{Related work.} Regarding the existence of polynomial kernels for \ILPFn only little is known. In general, parameterized by the number of variables, \ilpfn admits no polynomial kernelization unless \containment and the polynomial hierarchy collapses. This follows for example from the results of Dell and van Melkebeek~\cite{DellM10} regarding lower bounds for the compressibility of the satisfiability problem, since there is an immediate reduction from \sat to \ilpfn. Similarly, it follows also from earlier results of Dom et al.~\cite{DomLS09} who showed that \hittingset parameterized by the universe size admits no polynomial kernelization under the same assumption.

We note that both ways of excluding polynomial kernels for \ILPFn use reductions from problems with unbounded arity. Crucially, both \dhittingset and \dsat admit polynomial kernels of size roughly~$\Oh(n^d)$, where~$n$ is the number of elements and variables respectively, which can be obtained trivially by discarding duplicate sets or clauses, respectively. Surprisingly perhaps, the work of Dell and van Melkebeek~\cite{DellM10} shows that these bounds are tight, assuming \noncontainment, i.e., there are no reductions to size~$\Oh(n^{d-\epsilon})$ for any~$\epsilon>0$. We emphasize that this also implies the lower bound of \ILPFn since it can express, e.g., \hittingset with sets of unbounded size (exceeding any constant~$d$).

Motivated by these facts about the kernelization lower bound for \ILPFn and the existing straightforward polynomial kernels for \dhittingset and \dsat, we study the influence of arity on the existence of polynomial kernels for \ilpfn. Regarding the considered integer linear programs with constraints~$Ax\leq b$ this translates to~$A$ being~$r$-row-sparse, i.e., to have at most~$r$ nonzero entries in each row. (This is equivalent to requiring that each constraint has at most~$r$ variables with nonzero coefficients.)

\subparagraph{Our results.} We study \ILPFn for the case that the constraint matrix~$A$ is~$r$-row-sparse; we call this problem \RSILPFn (\rsilpfn). Note that~$r$ is a constant that is fixed as a part of the problem (it makes no sense to study~$r$ as an additional parameter since we already know that constraints involve at most all~$n$ variables, but already for \sat parameterized by the number of variables this is not enough to avoid a kernelization lower bound).

Our main result is that \rsilpfn admits no polynomial kernelization for any~$r\geq 3$, unless \containment.
Thus we see that unlike the simpler problems \dhittingset and \dsat, a restriction on the arity (or row-sparseness) is not enough to ensure a polynomial kernelization. For this result we give a cross-composition (introduced by Bodlaender et al.~\cite{BodlaenderJK11}; see Section~\ref{section:preliminaries}) from \clique to \rsilpfn. Concretely, we encode~$t$ instances of \clique into a single instance of \rsilpfn with parameter value bounded polynomially in the largest \clique instance plus~$\log t$, such that our obtained instance is \yes if and only if at least one of the \clique instances is \yes. This is presented in Section~\ref{section:sparselowerbound}. The lower bound can be seen to also apply to the case of parameterization by~$n+\log C$ where~$C$ is the largest absolute value of any coefficient (this refers to integer coefficients which can be obtained by scaling, or, alternatively, one could use the binary encoding size in place of~$\log C$); this is interesting since an ILP with~$n$ variables and~$m$ constraints can be trivially encoded in space~$\Oh(nm\log C)$.

Unlike other proofs via compositions or cross-compositions, the parameterization by the number of variables combined with the row-sparseness restriction prevent many standard tricks. For example, without the row-sparseness we could simply encode the selection of an instance number of one of the~$t$ \clique instances. Then we could add constraints that encode all the edges of the input graphs, but which are only valid when the binary encoding of the instance number matches the constraint. Unfortunately, this involves constraints with~$\Oh(\log t)$ variables.\footnote{We can emulate a few such constraints by use of auxiliary variables, but we cannot afford to do this for the constraints corresponding to all~$t$ instances.} (Of course without row-sparseness, a lower bound is known already.) Similarly, if we could use~$t$ slack variables we could very easily control the constraints and have only those for a single instance of \clique be relevant; however, we cannot afford this.

Our solution goes by using a significantly larger domain for the variables that encode the selection of a clique in one of the~$t$ input graphs. We use a variable~$s$ for the instance number, and add (linear) constraints that enforce~$r=s^2$. This permits us to use indicator variables for the desired clique whose feasible values depend quadratically on the chosen instance number. Accordingly, we can arrange the constraints for the edges of all input graphs~$G_i$, such that they intersect this feasible region when~$i=s$. In this way, depending on~$s$, only the constraints from one instance will restrict the choice of values for the indicator variables (beyond the restriction imposed directly by~$s$ and~$r=s^2$).

Complementing our lower bound, and recalling the large domain required for the construction, we analyze the effect of the maximum variable range on the preprocessing. It turns out that we can efficiently reduce row-sparse ILPs of form~$Ax\leq b$ to a size that is polynomial in~$n+d$, where~$n$ is the number of variables and~$d$ is the maximum range of any variable. In other words, \RSILPF admits a polynomial kernelization with respect to the combined parameter~$n+d$, or when~$d$ is polynomially bounded in~$n$; this is showed in Section~\ref{section:sparsedomainboundedkernel}. Together our upper and lower bound show that the existence for \RSILPF depends strongly on the permitted range for the variables. Note that, our lower bound proof (Section~\ref{section:sparselowerbound}) allows the conclusion that parameterization by~$n+\log d$ is not enough to allow a polynomial kernelization: The maximum value of any variable is polynomial in~$t$, implying that~$\log d=\Oh(\log t)$ which suffices for a cross-composition (see Definition~\ref{definition:crosscomposition}). We emphasize that small range without row-sparseness does not suffice by the mentioned reductions from \sat and \hittingset.

Furthermore, let us point out that for the case of an ILP of form~$Ax=b$,~$x\geq 0$, with~$n$ variables, Gaussian elimination suffices to reduce the number of constraints to~$n$, but still leaves the remaining problem of reducing the size of the coefficients in order to obtain a polynomial kernelization. Note that, while in general there are trivial transformations between~$Ax\leq b$ and~$A'x'=b'$, going from~$Ax\leq b$ to~$A'x'= b'$ uses one slack variable per constraint and hence would increase our parameter (the number of variables) by the number of constraints; this would make any further reduction arguments pointless.

\section{Preliminaries}\label{section:preliminaries}

\subparagraph{Parameterized complexity and kernelization.} A \emph{parameterized problem} over some finite alphabet~$\Sigma$ is a language~$\P\subseteq\Sigma^*\times\N$. The problem~$\P$ is \emph{fixed-parameter tractable} if~$(x,k)\in\P$ can be decided in time~$f(k)\cdot(|x|+k)^{\Oh(1)}$, where~$f$ is an arbitrary computable function. A polynomial-time algorithm~$K$ is a kernelization for~$\P$ if, given input~$(x,k)$, it computes an equivalent instance~$(x',k')$ with~$|x'|+k'\leq h(k)$ where~$h$ is some computable function;~$K$ is a \emph{polynomial} kernelization if~$h$ is polynomially bounded (in~$k$). By relaxing the restriction that the created instance~$(x',k')$ must be of the same problem and allow the output to be an instance of any classical decision problem we get the notion of \emph{(polynomial) compression}.

For our lower bound proof we use the concept of an (\OR-)cross-composition of Bodlaender et al.~\cite{BodlaenderJK11} which builds on a series of earlier results~\cite{FortnowS11,BodlaenderDFH09,DellM10} that created a framework for ruling out polynomial kernelizations for certain problems.

\begin{definition}[\cite{BodlaenderJK11}] \label{polyEquivalenceRelation}
An equivalence relation~$\R$ on~$\Sigma^*$ is called a \emph{polynomial equivalence relation} if the following two conditions hold:\\
1. There is a polynomial-time algorithm that decides whether two strings belong to the same equivalence class (time polynomial in~$|x|+|y|$ for~$x,y\in\Sigma^*)$.\\
2. For any finite set~$S \subseteq \Sigma^*$ the equivalence relation~$\R$ partitions the elements of~$S$ into a number of classes that is polynomially bounded in the size of the largest element of~$S$.
\end{definition}

\begin{definition}[\cite{BodlaenderJK11}]\label{definition:crosscomposition}
Let~$L\subseteq\Sigma^*$ be a language, let~$\R$ be a polynomial equivalence relation on~$\Sigma^*$, and let~$\P\subseteq\Sigma^*\times\N$ be a parameterized problem. An \emph{\OR-cross-com\-position of~$L$ into~$\P$} (with respect to $\R$) is an algorithm that, given~$t$ instances~$x_1, x_2, \ldots, x_t \in \Sigma^*$ of~$L$ belonging to the same equivalence class of~$\R$, takes time polynomial in~$\sum _{i=1}^t |x_i|$ and outputs an instance~$(y,k) \in \Sigma^* \times \mathbb{N}$ such that:\\
1. The parameter value~$k$ is polynomially bounded in~$\max_i|x_i|+\log t$.\\
2. The instance~$(y,k)$ is \yes for~$\P$ if and only if \emph{at least one} instance~$x_i$ is \yes for~$L$.\\
We then say that~$L$ \OR-cross-composes into~$\P$.
\end{definition}

\begin{theorem}[\cite{BodlaenderJK11}]\label{theorem:crosscomposition}
If an NP-hard language~$L$ \OR-cross-composes into the parameterized problem~$\P$, then~$\P$ does not admit a polynomial kernelization or polynomial compression unless \containment and the polynomial hierarchy collapses.
\end{theorem}

\section{A kernelization lower bound for sparse ILP Feasibility}\label{section:sparselowerbound}

In this section we show our main result, namely that a restriction to row-sparse matrices is not enough to ensure a polynomial kernelization for \ILPF parameterized by the number of variables. The problem is defined as follows.

\parameterizedproblem{$r$-\textsc{Sparse Integer Linear Programming Feasibility}($n$) -- \rsilpfn}{An~$r$-row-sparse matrix~$A\in \Q^{m\times n}$ and a vector~$b\in\Q^m$.}{$n.$}{Is there a vector~$x\in\Z^n$ such that~$Ax\leq b$?}

To prove the kernelization lower bound for \rsilpf we give an \OR-cross-composition from the NP-hard \clique problem, i.e., a reduction of many \clique instances into a single instance of \rsilpf. The idea behind the construction is to use a fairly large domain in order to recycle the same variables for the constraints that correspond to many different instances.

As a first step we state two propositions which together allow us to ``compute'' the square of a variable inside an ILP, i.e., to add constraints such that some variable is exactly the square of another in all feasible solutions.

\begin{proposition}\label{proposition:multiplication}
Let~$s_i$,~$s_j$,~$s_{ij}$, and~$d_{ij}$ denote integer variables with range~$\{0,1\}$ each. Then any feasible assignment for~$s_{ij}=\frac{1}{2}(s_i+s_j-d_{ij})$ satisfies~$s_{ij}=s_i\cdot s_j$. Conversely, for any choice of~$s_i$,~$s_j$, and~$s_{ij}$ such that~$s_{ij}=s_i\cdot s_j$, there is a choice of~$d_{ij}\in\{0,1\}$ such that~$s_{ij}=\frac{1}{2}(s_i+s_j-d_{ij})$ holds.
\end{proposition}

\begin{proposition}\label{proposition:square}
Let~$s\in\{0,\ldots,t-1\}$ with~$t=2^\ell$ and let~$s_0,\ldots,s_{\ell-1}\in\{0,1\}$ denote the binary expansion of~$s$, i.e.,~$s=\sum_{i=0}^{\ell-1}2^is_i$. Then
\[
s^2=\sum_{i=0}^{\ell-1}\left(\sum_{j=0}^{\ell-1}2^{i+j}s_i\cdot s_j\right).
\]
\end{proposition}

Together the two propositions provide a way of forcing some variable in an ILP to take a value exactly equal to the square of another value. If~$s\in\{0,\ldots,t-1\}$ this requires~$\Oh(\log^2 t)$ auxiliary variables and~$\Oh(\log^2 t)$ constraints.
Now we will give our construction.

\begin{theorem}\label{theorem:sparselowerbound}
Let~$r\geq 3$ be an integer. The \rsilpf problem does not admit a polynomial kernelization or compression unless \containment and the polynomial hierarchy collapses.
\end{theorem}

\begin{proof}
We give an \OR-cross-composition from the NP-hard \clique problem. Let~$t$ instances of \clique be given. By a polynomial equivalence relation that partitions instances according to number of vertices and requested clique size it suffices to consider instances that ask for the same clique size~$k$ and such that each input graph has~$n$ vertices. We denote the instances~$(G_0,k),\ldots,(G_{t-1},k)$; for convenience, assume that all~$t$ graphs have the same vertex set~$V$ and edge sets~$E_i$ for~$i\in\{0,\ldots,t-1\}$. We will create a single instance of \RSILPFn that is \yes if and only if at least one instance~$(G_i,k)$ is \yes for \clique. Without loss of generality, we assume that~$t=2^\ell$; otherwise we could copy some instance sufficiently often (at most doubling the input size).

\subparagraph{Construction--essential part.} For the sake of readability we first describe the matrix~$A$ by writing down the constraints in a succinct way ignoring the sparsity requirement; there will be a small number of constraints on more than three variables which will be converted later. We also specify explicit ranges for the variables which can be enforced by the obvious constraints. Note that~$n$,~$t$,~$\ell$,~$k$,~$i$, and~$j$ are constants in the ILP;~$i$ and~$j$ are used in sums but the expansion of each sum is a constraint where~$i$ and~$j$ have constant values.

The first group of variables, namely~$s$ and~$s_0,\ldots,s_{\ell-1}$ serve to pick an instance number~$s\in\{0,\ldots,t-1\}$ and enforce the variables~$s_i$ to equal the binary expansion of~$s$.
\begin{align}
s&\in\{0,\ldots,t-1\}\label{range:s}\\
s_0,\ldots,s_{\ell-1}&\in\{0,1\}\label{range:si}\\
s&=\sum_{i=0}^{\ell-1}2^is_i\label{constraint:binaryexpansion}
\end{align}

Next we create a variable~$r$ and auxiliary variables~$s_{ij}$ and~$d_{ij}$ with the sole purpose of enforcing~$r=s^2$ but using only linear constraints.
\begin{align}
r&\in\{0,\ldots,(t-1)^2\}\label{range:r}\\
s_{ij},d_{ij}&\in\{0,1\} && \mbox{for all~$i,j\in\{0,\ldots,\ell-1\}$} \label{range:sijdij}\\
s_{ij}&=\frac{1}{2}(s_i+s_j-d_{ij}) && \mbox{for all~$i,j\in\{0,\ldots,\ell-1\}$} \label{constraint:multiplication}\\
r&=\sum_{i=0}^{\ell-1}\left(\sum_{j=0}^{\ell-1}2^{i+j}s_{ij}\right)\label{constraint:computesquare}
\end{align}

We introduce variables~$y_v$ for all~$v\in V$ which will encode a~$k$-clique in instance~$s$. These variables are restricted to take one of two values that depend on~$s$ in a quadratic way (using~$r=s^2$; recall that~$t$ is a constant).
\begin{align}
y_v&\leq 2ts-r+2\quad\quad\mbox{for all~$v\in V$}\label{constraint:upperbound}\\
y_v&\geq 2ts-r+1\quad\quad\mbox{for all~$v\in V$}\label{constraint:lowerbound}
\end{align}
That is, we restrict~$y_v$ to~$y_v\in\{2ts-r+1,2ts-r+2\}\subseteq\{0,\ldots,2t^2\}$.

Now we get to the central piece of the ILP, namely the constraints which will enforce the non-edges of the graph~$G_s$. However, we of course need to add those constraints for all input graphs~$G_i$. It is crucial that only the constraints for~$i=s$ have an effect on the~$y$-variables (beyond the restriction already imposed by~\cref{constraint:upperbound} and \cref{constraint:lowerbound}). We add the following for all~$\{u,v\}\subseteq V$ and instance numbers~$i\in\{0,\ldots,t-1\}$ if~$\{u,v\}$ is not an edge of~$G_i$.
\begin{align}
\mbox{if~$\{u,v\}\notin E_i$ then}\quad\quad y_u+y_v&\leq 4\cdot(t-i)\cdot s+2i^2+3\label{constraint:nonedge}
\end{align}

Finally, we take the sum over all~$y_v$, deduct~$n$ times the minimum value~$2ts-r+1$ and check that this is at least as large as the specified target value~$k$.
\begin{align}
\left(\sum_{v\in V}y_v\right)-n\cdot(2ts-r+1)\geq k \label{constraint:targetvalue}
\end{align}

This completes the essential part of the construction. Formally we still need to convert all constraints into form~$Ax\leq b$ and to use only three variables in each constraint. However, the proof will be given regarding the more accessible constraints stated above. 

\subparagraph{Construction--formal part.}
We use~$x$ to refer to the vector of all variables used above, e.g.,~$x=(s,s_0,\ldots,s_{\ell-1},r,s_{00},\ldots,s_{\ell-1,\ell-1},d_{00},\ldots,d_{\ell-1,\ell-1},y_{v_1},\ldots,y_{v_n})$.
Thus, at this point, we use~$1+\ell+1+2\cdot\ell^2+n\in\Oh(n+\ell^2)=(n+\log t)^{\Oh(1)}$ variables.

To formally complete the construction one now needs to translate all constraints to form~$Ax\leq b$. Furthermore, using auxiliary variables, one needs to convert this to~$A'x'\leq b'$ such that~$A'$ has at most three non-zero entries in each row. It is clear that all range constraints, namely \cref{range:s}, \cref{range:si}, \cref{range:r}, and \cref{range:sijdij} can be expressed by two linear inequalities with one variable each. Also the constraints \cref{constraint:upperbound}, \cref{constraint:lowerbound}, and \cref{constraint:nonedge} need no further treatment since they are already linear inequalities with at most three variables each (that is, it suffices to rearrange them to have all variables on one side when transforming to~$Ax\leq b$).

For the remaining constraints, namely~\cref{constraint:binaryexpansion},~\cref{constraint:multiplication},~\cref{constraint:computesquare}, and~\cref{constraint:targetvalue} we need to use auxiliary variables to replace them by small sets of linear inequalities with at most three variables each. We sketch this for~\cref{constraint:binaryexpansion}, which requires expressing a sum using partial sums. We introduce~$\ell$ new variables~$z_0,\ldots,z_{\ell-1}$ and replace~$s=\sum_{i=0}^{\ell-1}2^is_i$ as follows; the intuition is that~$z_j=\sum_{i=0}^{j}2^is_i$.
\begin{align*}
z_0-s_0\leq 0 &&
-z_0+s_0\leq 0\\
z_i-z_{i-1}-2^i s_i\leq 0 &&
-z_i+z_{i-1}+2^i s_i \leq 0 &&\quad\quad\quad\mbox{for~$i\in\{1,\ldots,\ell-1\}$}\\
s-z_{\ell-1}\leq 0 &&
-s+z_{\ell-1}\leq 0
\end{align*}

We use~$\ell$ variables for constraint~\cref{constraint:binaryexpansion},~$\ell^2$ variables for constraints~\cref{constraint:multiplication},~$\ell^2$ variables for constraint~\cref{constraint:computesquare}, and~$n+2$ variables for constraint~\cref{constraint:targetvalue}. Altogether we use~$\Oh(n+\ell^2)=\Oh(n+\log^2 t)$ additional variables. In total our ILP uses~$\Oh(n+\log^2 t)=\Oh((n+\log t)^{\Oh(1)})$ variables, which is consistent with the definition of a cross-composition (polynomial in the largest input instance plus the logarithm of the number of instances).

\subparagraph{Completeness.} To show correctness, let us first assume that some instance~$(G_\i,k)$ is \yes for \clique, and let~$C\subseteq V$ be some~$k$-clique in~$G_\i$. We will determine a value~$x'=x'(\i,C)$ such that~$A'x'\leq b'$ (this is the system obtained by transforming all constraints to inequalities in at most three variables). Again, for clarity, we will simply pick values only for all variables used in the succinct representation (i.e., all variables occurring in~\cref{range:s}--\cref{constraint:targetvalue}) and check that all (in-)equalities are satisfied. It is obvious how to extend this to the auxiliary variables that are required for formally writing down all constraints as~$A'x'\leq b'$.

First of all, we set~$s=\i\in\{0,\ldots,t-1\}$ and set the variables~$s_0,\ldots,s_{\ell-1}\in\{0,1\}$ such that they match the binary expansion of~$s$. Clearly, this satisfies constraint~$\cref{constraint:binaryexpansion}$ as well as the range of each encountered variable. It follows from Proposition~\ref{proposition:multiplication} that we can set~$s_{ij}=s_i\cdot s_j\in\{0,1\}$ and also find feasible values for all~$d_{ij}$ such that all constraints~\cref{constraint:multiplication} are satisfied. Hence, by Proposition~\ref{proposition:square} we can set~$r=s^2$ while satisfying constraint~\cref{constraint:computesquare}.

Now, let us assign values to variables~$y_v$ for~$v\in V$ as follows
\begin{align*}
y_v=\begin{cases}
    2ts-r+2 & \mbox{if,~$v\in C$}\\
    2ts-r+1\ & \mbox{if~$v\notin C$.}
    \end{cases}
\end{align*}
It is easy to see that this choice satisfies both constraints~\cref{constraint:lowerbound} and~\cref{constraint:targetvalue}, since~$|C|=k$.

Finally, we have to check that the (non-)edge constraints~\cref{constraint:nonedge} are satisfied for all~$i\in\{0,\ldots,t-1\}$ and all edges~$\{u,v\}$. There are two cases, namely~$i=\i$ and~$i\neq \i$, i.e., we have to satisfy constraints for~$G_\i$ (using the fact that~$C$ is a clique) but also constraints created for graphs~$G_i$ with~$i\neq \i$. 

Let us first consider the case~$i\neq \i$; concretely, we take the maximum value for~$y_u+y_v$, namely~$2\cdot(2ts-r+2)$,  and compare it to the value of constraint~\cref{constraint:nonedge}, namely~$4\cdot(t-i)\cdot s+2i^2+3$, using that~$r=s^2$ and~$s=\i$:
\begin{align*}
&& 4\cdot(t-i)\cdot s+2i^2+3&\geq 2\cdot(2ts-r+2)\\
\Leftrightarrow&& 4ts-4is+2i^2+3&\geq 4ts-2s^2+4\\
\Leftrightarrow&& 2s^2-4is+2i^2-1&\geq 0\\
\Leftrightarrow&& 2(s-i)^2-1&\geq 0.
\end{align*}
Since~$s=\i$ the last inequality holds if~$i\neq \i$, which is exactly what we assumed. Thus all non-edge constraints for graphs~$G_i$ with~$i\neq \i$ are satisfied.

We now consider the non-edge constraints for~$G_\i$. We compute the difference between the bound of constraint~\cref{constraint:nonedge} and the minimum value of~$y_u+y_v$, namely~$2\cdot(2ts-r+1)$, to check that our assignment to~$y$-variables is feasible. Note that~$r=s^2$ and~$s=\i=i$:
\begin{align*}
&\quad(4\cdot(t-i)\cdot s+2i^2+3)-2\cdot (2ts-r+1) =4ts-4is+2i^2+3-4ts+2s^2-2 = 1.
\end{align*}
Thus, if~$\{u,v\}\notin E_\i$ then at most one of~$y_u$ and~$y_v$ can take value~$2ts-r+2$ without violating constraint~\cref{constraint:nonedge}. Otherwise, if~$\{u,v\}\in E_\i$, then, from the perspective of this edge, both variables may take value~$2ts-r+2$. Clearly, this is consistent with our assignment to the~$y$-variables, since the larger value~$2ts-r+2$ is assigned to all variables that correspond to the vertices of the~$k$-clique~$C$.

\subparagraph{Soundness.} For soundness, let us assume that we have a feasible solution~$x'$ such that~$A'x'\leq b'$. Again, we consider only the variables of constraints~\cref{range:s}--\cref{constraint:targetvalue}. Recall that~$s\in\{0,\ldots,t-1\}$. We claim that the graph~$G_s$ must have a clique of size at least~$k$.

Observe that all variables~$y_v$ for~$v\in V$ have value~$2ts-r+2$ or~$2ts-r+1$ in~$x$ due to constraints~\cref{constraint:upperbound} and~\cref{constraint:lowerbound}. We define a vertex subset~$C\subseteq V$ by stating that it contains exactly those vertices~$v$ with~$y_v=2ts-r+2$. The goal is to show that~$C$ is a clique in~$G_s$.

As for the converse direction, feasible solutions are required to have~$r=s^2$, which follows from Propositions~\ref{proposition:multiplication} and~\ref{proposition:square}; note that obviously the variables~$s_0,\ldots,s_{\ell-1}$ need to equal the binary expansion of~$s$ due to constraint~\cref{constraint:binaryexpansion}.

Now, we consider the non-edge constraints~\cref{constraint:nonedge} for~$G_s$ and compare them to the lower bound of~$2ts-r+1$ for variables~$y_v$; we already did this computation earlier, again we have~$r=s^2$ and~$s=i$:
\begin{align*}
4\cdot(t-i)\cdot s+2i^2+3-2\cdot(2ts-r+1)=1.
\end{align*}
Hence, for every non-edge~$\{u,v\}$ of~$G_s$ among~$y_u$ and~$y_v$ at most one of the two variables can take the larger value~$2ts-r+2$. Therefore, when~$y_u=y_v=2ts-r+2$, then~$\{u,v\}$ is an edge of~$G_s$. Thus,~$C$ is a clique in~$G_s$.
It follows from~$y_v\in\{2ts-r+1,2ts-r+2\}$ that constraint~\cref{constraint:targetvalue} enforces that~$y_v=2ts-r+2$ for at least~$k$ vertices~$v\in V$. Therefore,~$C$ is of size~$k$. This completes the \OR-cross-composition from \clique.

By Theorem~\ref{theorem:crosscomposition}, \RSILPFn has no polynomial kernelization unless \containment and the polynomial hierarchy collapses~\cite{BodlaenderJK11}.
\end{proof}

The cross-composition in the proof of Theorem~\ref{theorem:sparselowerbound} uses variables of range polynomial in~$t$ and coefficients of absolute value bounded polynomially in~$t$. We will discuss the aspect of variable range in the following section. The size of the coefficients is also interesting since an ILP with integer coefficients (like the one we create) can be easily encoded in space~$\Oh(nm\log C)$ where~$C$ is the absolute value of the largest coefficient. As the given cross-composition has~$C=t^{\Oh(1)}$ we see that space polynomial in~$\log t$ suffices, and hence the lower bound applies also to \RSILPF{($n+\log C$)}; regarding parameters~$n$,~$m$, and~$\log C$ this is a maximal negative case since parameterization by~$n+m+\log C$ trivially gives a polynomial kernel (by the mentioned encoding). Put differently, the obstacle established in the lower bound proof is the large number of coefficients; coefficients of \emph{value} polynomial in~$t$ are required to make this work, but it is not their \emph{encoding size} that is the obstacle for polynomial kernels.

\section{A polynomial kernelization for sparse ILP with bounded range}\label{section:sparsedomainboundedkernel}

We have seen that for \RSILPFn there is no polynomial kernelization unless \containment. The proof relies strongly on having variables of high range in order to encode the constraints of~$t$ instances of~$\clique$. 
It is natural to ask, whether a similar result can be proven when the maximum range of any variable is small, e.g., polynomial in the number of variables. We show that this is not the case by presenting a polynomial kernelization for the variant where the maximum range is an additional parameter.
The problem is defined as follows.

\parameterizedproblem{\RSBILPFnd}{An~$r$-row-sparse matrix~$A\in \Q^{m\times n}$ and a vector~$b\in\Q^m$.}{$n+d.$}{Is there a vector~$x\in\{0,\ldots,d-1\}^n$ such that~$Ax\leq b$?}

Note that we restrict to the seemingly special case where each variable is not only restricted to~$d$ different consecutive values, but in fact all variables must take values from~$\{0,\ldots,d-1\}$. It can be easily checked that this is as general as allowing any~$d$ consecutive integers, since we could shift variables to range~$\{0,\ldots,d-1\}$ without changing feasibility (by changing~$b$).

\begin{theorem}
$r$-\textsc{Sparse Bounded Integer Linear Programming Feasibility}($n,d$) admits a polynomial kernelization with size~$\Oh(n^r\cdot d^r\cdot \log nd)$.
\end{theorem}

\begin{proof}
We assume that~$r\geq 3$ since otherwise the problem can be solved in time~$\Oh(m\cdot d)$ by work of Bar-Yehuda and Rawitz~\cite{BarYehudaR01} and the theorem follows trivially. Recall that for~$r\geq 3$ the problem is NP-hard by a reduction from~$3$-\sat.

The kernelization works by considering all choices of~$r$ of the~$n$ variables and replacing the constraints (i.e., inequalities) in~$Ax\leq b$ which contain only those variables. The starting observation is that there are~$d^r$ choices of picking values for~$r$ variables, and the considered constraints prevent some of those from being feasible. It can be efficiently checked which of the~$d^r$ assignments are feasible. For each infeasible point~$P=(p_1,\ldots,p_r)$ we show how to give a small number of constraints that exactly exclude this point. Together, all those new constraints have the same effect as the original ones, allowing the latter to be discarded.

Let~$x_1,\ldots,x_r$ be any~$r$ of~$n$ variables and let~$\hat{P}$ denote the set of all points~$P=(p_1,\ldots,p_r)$ that are infeasible for constraints only involving~$x_1,\ldots,x_r$. (Note that the whole ILP might be infeasible, but locally we only care for an equivalent replacement of the constraints.) We show constraints that enforce~$(x_1,\ldots,x_r)\neq(p_1,\ldots,p_r)$:
\begin{align}
\forall i\in\{1,\ldots,r\}:&& x_i&=p_i+s_i-d\cdot t_i && s_i\in\{0,\ldots,d-1\}, t_i\in\{0,1\}\label{constraint:checkequal}\\
&&\sum_{i=1}^r s_i&\geq 1 \label{constraint:checkdummies}
\end{align}
This requires~$2r$ variables and~$r+1$ constraints; a few more variables and constraints are required to transform the constraints into an equivalent set of inequalities with at most~$r$ variables each: For constraint~\cref{constraint:checkdummies} it suffices to flip the sign since it is already an inequality on~$r$ variables. For constraints~\cref{constraint:checkequal} we can replace each equality by two equalities using a new auxiliary variable (in fact this is only needed when~$r=3$) and replacing both equalities in turn by two inequalities. We use~$3r$ variables and~$4r+1$ constraints total. Note that all coefficients have values in~$\{-1,0,1,d\}$ and can be encoded by~$\Oh(\log d)$ bits (in fact two bits suffice easily for four values).

Again, we will argue correctness on the more succinct representation, i.e., on~\cref{constraint:checkequal} and~\cref{constraint:checkdummies}.

Assume first that~$(x_1,\ldots,x_r)=(p_1,\ldots,p_r)$. Thus~$0=x_i-p_i=s_i-d\cdot t_i$, which implies that~$s_i=t_i=0$ (taking into account the domains of~$s_i$ and~$t_i$) for all~$i$. Thus constraint~\cref{constraint:checkdummies} is violated, making~$(x_1,\ldots,x_r)=(p_1,\ldots,p_r)$ infeasible.
On the other hand, if~$(x_1,\ldots,x_r)\neq (p_1,\ldots,p_r)$, then there is a position~$j$ with~$x_j\neq p_j$. It follows that~$0<|x_j-p_j|<d$ (due to the range of~$x_j$) which in turn implies that~$s_j\neq 0$ since the contribution of~$d\cdot t_j$ to the equality is a multiple of~$d$. Thus constraint~\cref{constraint:checkdummies} is fulfilled.

It follows that we are able to add constraints which exclude any desired point for~$x_1,\ldots,x_r$. Let us complete the proof. Clearly, if a vector~$x$ fulfills~$Ax\leq b$ then any choice of~$r$ variables from~$x$ fulfills all constraints that contain only these variables. This in turn means that those variables avoid the points that are excluded by the constraints, which implies that they satisfy all our new constraints (since avoiding those points is all that is needed).

Conversely, assume that a vector~$x$ fulfills all new constraints and hence any choice of~$r$ variables avoids all forbidden points. Since any of the original constraints contains at most~$r$ variables, it comes down to forbidding some set of points. Since~$x$ fulfills our new constraints it also avoids all infeasible points for~$Ax\leq b$. Thus,~$x$ satisfies also all original constraints.

Summarizing, we are able to replace all constraints by new constraints with small coefficients, which have the same outcome. Clearly the computations can be performed in polynomial time (the input size dominates~$n$,~$m$, and the encodings of all coefficients in~$A$ and~$b$). Since for any~$r$ variables there are at most~$d^r$ infeasible points, we need at most~$(4r+1)\cdot d^r\cdot \binom{n}{r}=\Oh(d^r\cdot n^r)$ constraints and~$3r\cdot d^r\cdot n^r=\Oh(d^r\cdot n^r)$ variables. The generated equivalent instance can be encoded by~$r\cdot\Oh(\log (d^r\cdot n^r))\cdot \Oh(d^r\cdot n^r)=\Oh(d^r\cdot n^r\cdot \log dn)$ bits, by encoding each constraint (on~$r$ variables) as the binary encoded names of the variables with nonzero coefficients followed by the values of the coefficients. 
\end{proof}

\section{Conclusion}

We prove that the existence of polynomial kernels for \RSILPF with respect to the number~$n$ of variables depends strongly on the maximum range of the variables. If the range is unbounded, then there is no polynomial kernelization under standard assumptions. Otherwise, if the range of each variable is polynomially bounded in~$n$ then we establish a polynomial kernelization. This holds also for the more general case of using the maximum range as an additional parameter.

Future work will be directed at more restricted cases of ILPs in order to obtain more positive kernelization results. Similarly, structural parameters of ILPs seem largely unexplored.

\bibliographystyle{plain}
\bibliography{ilpkernels_arxiv}
\end{document}